\documentclass{article}
\usepackage[french,english]{babel}
\usepackage[utf8]{inputenc}
\usepackage[T1]{fontenc}
\usepackage{csquotes}
\usepackage{amsmath}
\usepackage{mathrsfs}
\usepackage{amssymb}
\usepackage{amsfonts}
\usepackage{dsfont}
\usepackage{amsthm}
\usepackage{latexsym}
\usepackage{graphicx}
\usepackage{float}

\newtheorem{theorem}{Theorem}[section]
\newtheorem{lemme}[theorem]{Lemma}
\newtheorem{proposition}[theorem]{Proposition}
\newtheorem{corollary}[theorem]{Corollary}
\newtheorem{definition}[theorem]{Definition\rm}
\newtheorem{remark}{Remark}
\newtheorem{example}{Example}

  \usepackage[
    backend=biber,
    style=bwl-FU,
  ]{biblatex}
\usepackage{mathtools}
\usepackage{graphicx}
\usepackage{tikz}
\usepackage{tikz-cd}
\usetikzlibrary{matrix,arrows,decorations.pathmorphing}

\DeclareMathOperator*{\E}{\mathbb{E}}
\usepackage{authblk}

\title{Risk Quantization by Magnitude and Propensity}
\author[1]{Olivier P. Faugeras\thanks{Corresponding author}}
 \author[2]{Gilles Pagès}
\date{\today}
\providecommand{\keywords}[1]
{
  \small	
  \textbf{\textbf{Keywords:}} #1
}
\affil[1]{\small Toulouse School of Economics,
 Universit\'{e} Toulouse 1 Capitole, 
  1, Esplanade de l'Universit\'{e}. Bureau T106. 31080 Toulouse Cedex 06, France. 
  \texttt{ olivier.faugeras@tse-fr.eu}}

\affil[2]{Laboratoire de Probabilités, Statistique et Modélisation, UMR 8001, Campus Pierre et Marie Curie, Sorbonne Université, case 158, 4, pl. Jussieu, F-75252 Paris Cedex 5, France.
\texttt{gilles.pages@upmc.fr}
}

\begin{document}

\maketitle
\begin{abstract}
We   propose a novel approach in the assessment of a   random risk variable $X$  by introducing   magnitude-propensity risk measures $(m_X,p_X)$. This bivariate measure intends to account for the dual aspect of risk, where the magnitudes $x$ of   $X$ tell how hign are the losses incurred, whereas the probabilities $P(X=x)$ reveal how often one has to expect to suffer such losses.  The basic idea is to simultaneously quantify both the severity $m_X$ and the propensity $p_X$ of the real-valued risk $X$. This is to be contrasted with traditional univariate risk measures, like VaR or Expected shortfall, which typically conflate both effects. 
In its simplest form, $(m_X,p_X)$ is obtained by mass transportation in Wasserstein metric of the law $P^X$ of $X$ to a two-points $\{0, m_X\}$ discrete distribution with mass $p_X$ at $m_X$. The approach can also be formulated as a constrained optimal quantization problem. 
 This allows for an informative comparison of risks on both the magnitude and propensity scales. Several examples illustrate the proposed approach.

\end{abstract}

\section{Introduction and outline}
The evaluation and comparison of risks are basic tasks of risk analysis. The usual view in Insurance mathematics is to evaluate an univariate risk\footnote{We take the Insurance Mathematics convention, where the non-negative values of $X$ stands for the  loss incurred} $X$ by a risk measure $\rho(X)$, which can be thought of as a one-sided deterministic univariate summary of the random variable $X$. Risks $X$, $Y$ are then compared through their respective risk measures $\rho(X)$, $\rho(Y)$.

The starting point of this paper is the basic realization that risk, as a random variable, is intrinsically a bivariate phenomenon:   magnitudes (loss amounts) occurs with  given propensities (or probabilities). Hence, it appears inescapable that a risk measure, as a single univariate quantity on the magnitude scale, will conflate both effects, thus giving a somehow blurred representation of the risk borne by the random variable $X$. 
It would then be of interest to quantify risk on both the magnitude and propensity scales.
The purpose of this paper is to propose such a simultaneous quantification.

The proposed approach is based on the following  idea: as mentioned above, a risk measure $\rho(X)$ can be viewed as a deterministic proxy of the  random risk $X$. Distributionally speaking, it can be thought of as a Dirac measure $\delta_{\rho(X)}(.)$ at $\rho(X)$: this distribution  gives full propensity one at the magnitude level $\rho(X)$. Therefore, 
in order to quantify the risk with a varying propensity $p_X$, it  makes sense to look for an  approximate of the distribution of $X$ by a mixture of two Dirac, a Dirac at location zero with weight $1-p_X$,  and a Dirac at location $m_X$, with weight $p_X$. This proxy distribution thus encodes the magnitude and propensity effect of the risk borne by $X$, through the pair $(p_X,m_X)$. Mathematically, this problem of approximating   distributions is carried out by mass transportation in Wasserstein metric. Optimal transportation to  discrete measure also corresponds to the problem of optimal quantization, well-known in the Engineering and Signal Processing literature. Hence, the proposed approach  to quantify risk on both the magnitude and propensity scales amounts to a special, constrained, optimal quantization problem. 

The outline of the paper is as follows:
In Section \ref{section1}, we motivate the magnitude~-propensity approach to risk measures. 
As a new paradigm to risk evaluation,  the proposed approach needs a careful and detailed exposition of its main idea.
We first argue and give some evidence of this magnitude-propensity duality of risk, using (intentionally) simplistic examples to make our point clear. We then show how some classical risk measures typically mixes both effects and explain why it would be desirable to quantify risk on both the magnitude and propensity scales. 
A risk measures can usually  be derived as a M-functional, i.e. as a statistical parameter which is a solution of a problem of minimization of some expected loss. We next show how these expected loss minimization problems can be embedded into (degenerate) optimal transportation problems, i.e. towards a Dirac distribution. Eventually, we propose our definition of the magnitude-propensity pair $(p_X,m_X)$, as hinted above,  and make the connection to optimal quantization.

Section \ref{section2} is a theoretical study of the basic idea. We show how to compute the optimal $(m_X,p_X)$, both as direct minimization problem and as an optimal quantization problem. To that purpose we briefly recall the main facts about quantization theory. We provide existence, characterization and (partial) uniqueness results. We discuss some basic properties of the obtained magnitude-propensity pairs and  study some example distributions. 

In Section \ref{sec:section3}, we give some numerical illustrations of the basic idea. We explain how the  optimal $(m_X,p_X)$ gives rise to magnitude-propensity plots, which allow for an informative comparison of risks on both the  magnitude and propensity scale. We show how to do such comparisons  for distributions with or without explicit formulas for $(m_X,p_X)$. Empirically, the optimal magnitude-propensity pair can be computed by a fixed point algorithm, which is akin to Lloyd's algorithm in optimal quantization. The methodology is illustrated on a real data set of insurance losses. 

Eventually, we give a conclusion of the main results in Section \ref{sec:conclusion} and add some perspective for further research about possible variants and extensions of the basic approach. 

\section{The magnitude-propensity approach to risk measures}\label{section1}
\subsection{Risk is  intrinsically a bivariate propensity-magnitude random phenomenon}\label{section1-1}

As a primary concept of Insurance Mathematics, it is somehow difficult to give a precise, deductive definition of risk in terms of more primitive concepts. This explains why, e.g.  \cite{Novak2012} p. 224, states that   ``There is currently no consensus concerning the meaning of the word risk''. Hence, risk is often given an implicit or semantic definition in textbooks, as in \cite{Novak2012} p. 191:  ``Risk is a possibility of an undesirable event. Though such an event is rare, its magnitude can be devastating.'' 

The latter description of risk encompasses the duality magnitude-propensity of random variables: intuitively, a claim can be ``risky'' because losses may happen often, i.e. with a ``high'' propensity,  possibly with (relatively) ``small'' magnitudes,  or because a catastrophe of very large magnitude may happen, albeit with a (relatively) ``low'' propensity. We mention that this dual nature of risk is explicit in the Engineering literature (see e.g. \cite{bedford_cooke_2001}),  where it is  summarized by the semantic formula ``risk=uncertainty+damage'' in \cite{kaplan1981quantitative}.

Risk due to high magnitudes and risk due to high propensities manifests itself on several levels: intrinsically for a single risk variable, relatively in the comparison of two risks, and in the stochastic order itself. The best way to  illustrate our point is to give examples stripped to their simplest form, i.e. for two-points discrete measures.  Hence, we have the following intentionally simplistic examples:
\begin{itemize}
    \item \textbf{Intrinsic magnitude-propensity aspect for a single random risk:}
    
   	Let $X_1$ s.t. $P(X_1=1000)=0.1$, $P(X_1=0)=0.9$, while $X_2$ is s.t. $P(X_2=200)=0.5$, $P(X_2=0)=0.5$.  $X_1$ and $X_2$ have the same mean $\E [X_1]=\E [X_2]=100$.  $X_1$  has a ``high'' magnitude risk $1000$ of ``low'' propensity $0.1$ while $X_2$ has a ``small'' magnitude risk $200$ of ``high'' propensity $0.5$, see Figure \ref{pic-fig1}.

\begin{figure}[H]
	\centering
	\includegraphics[scale=0.6]{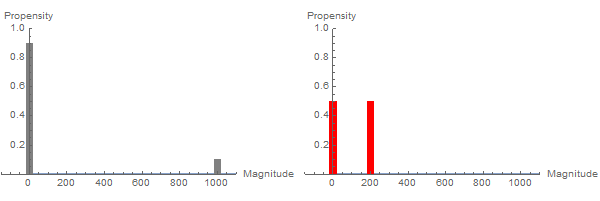}
	\caption{The dual nature of risk: $X_1$ has high magnitude risk (left) vs $X_2$ has high propensity risk (right). Yet, $\E[X_1]=\E[X_2]$.}
	\label{pic-fig1}
\end{figure}

    \item \textbf{Relative magnitude-propensity effect for the comparison of risks:}
    
    When one considers a pair  $(X_1,X_2)$ of risk variables and  regards risk as a relative property of one r.v. $X_1$ w.r.t. the other r.v. $X_2$, the magnitude-propensity duality   manifests itself in the comparison of risks and in the ordering structure of probability measures. This is illustrated in the following (also intentionally simplistic) examples:
	
	\begin{example}
	\label{example1}
	   Let $X_1$ s.t. $P(X_1=0)=0.9$, $P(X_1=100)=0.1$, while $X_2$ is s.t. $P(X_2= 0)=0.9$, $P(X_2=1000)=0.1$. Then, $X_1$ and $X_2$ have same propensities, but $X_2$ is more risky than $X_1$ due to a difference in magnitudes, see Figure \ref{pic-fig2}.
	\end{example}
	   
	\begin{example}
	\label{example2}
	 	 	Let $X_1$ s.t. $P(X_1=0)=0.95$, $P(X_1=1000)=0.05$, while $X_2$ is s.t. $P(X_2=0)=0.5$, $P(X_2=1000)=0.5$, see Figure \ref{pic-fig3}. Then, $X_1$ and $X_2$ have same magnitudes, but $X_2$ is more risky than $X_1$ due to a difference in propensities, see Figure \ref{pic-fig3}.
	\end{example}
	
		\begin{figure}[H]
	\centering
	\includegraphics[scale=0.6]{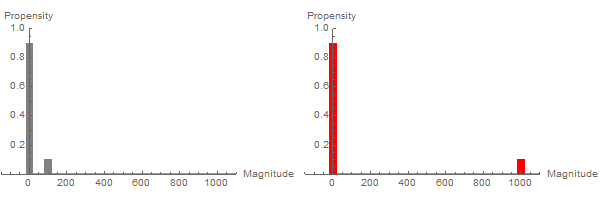}
	\caption{Comparison of risks for Example \ref{example1}: $X_2$ (right) is riskier than $X_1$ due to a difference of  magnitudes}
	\label{pic-fig2}
    \end{figure}
    
	\begin{figure}[H]
	\centering
	\includegraphics[scale=0.6]{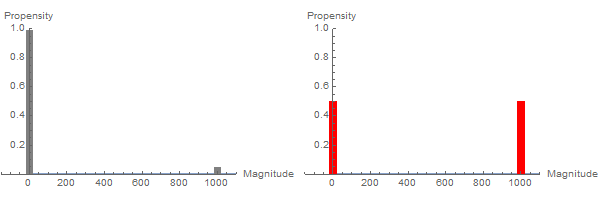}
	\caption{Comparison of risks for Example \ref{example2}: $X_2$ (right) is riskier than $X_1$ due to a difference of  propensities}
	\label{pic-fig3}
    \end{figure}

\newpage
    \item \textbf{Magnitude-propensity effect in the stochastic order:}
    
    In both cases of Examples \ref{example1} and \ref{example2}, one has 
\begin{equation}
    F_{X_1}(x)\ge F_{X_2}(x), \quad\forall x\in\mathbb R^+,
    \label{eq:stochorder}
\end{equation} 
so that 
$X_1\prec_{st} X_2$, where $\prec_{st}$ stands for the usual stochastic order. Yet,  (\ref{eq:stochorder})
does not distinguish between   a horizontal shift of c.d.f.s due to a difference in  magnitudes (Figure \ref{pic-fig4} left), and a vertical shift of c.d.f.s due to a difference of propensities (Figure \ref{pic-fig4} right).
\begin{figure}[H]
  \includegraphics[scale=0.6]{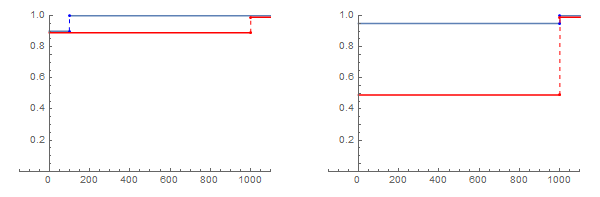}
	\caption{C.d.f.s of $X_1$ (blue), $X_2$ (red). Left: Example \ref{example1}, difference of magnitudes (horizontal shift). Right: Example \ref{example2}, differences of propensities (vertical shift).}
	\label{pic-fig4}
    \end{figure}
\end{itemize}

\begin{remark}
\label{rem:Galois}
From the mathematical standpoint, let us remark that 
this magnitude-propensity duality of risk 
  can be formalized mathematically by the concept of  a Galois connection between  the magnitude and probabilities spaces, considered as two ordered spaces,  $(\mathbb R,\le)$ and $([0,1],\le)$.
  Indeed, denote by  $F_X$, resp. $Q_X$, the cumulative distribution function, resp. the quantile function, of $X$. Then, $(F_X,Q_X)$ forms a Galois connection   between the magnitudes space $(\mathbb R,\le)$ and the propensity space  $([0,1],\le)$, i.e. for all $x\in\mathbb R$ and $t\in (0,1)$,
	$$
	t\le F_X(x) \Leftrightarrow Q_X(t)\le x.
	$$
	See \cite{FR2017} for  details and a generalization of the concept of quantile to the multivariate case.
\end{remark}

\subsection{Risk measures as  univariate deterministic proxies for a random risk}\label{section1-2}

The  main paradigm to the  assessment of risk in Insurance and Financial Mathematics is based on the use of risk measures 
\begin{equation}
\rho:\mathscr X\to \mathbb R^+,    
\label{eq:risk measure}
\end{equation}
 where $\mathscr X$ is a space of non-negative measurable random variables modeling an insurance claim. These risk measures quantify risk on a magnitude scale, by converting a random loss $X$ into a deterministic certainty equivalent $\rho(X)$. The latter  can then be used for ordering different risks,  and for decision making purposes, like setting the premium for covering the risk $X$, see e.g.  \cite{r2013}.
The resulting risk measure should typically satisfy some desirable properties (coherent risk measure), like translation invariance, monotonicity,  etc., see e.g. \cite{Artzner1999, Follmer}.

The last decades have seen a multiplication of risk measures. Among the numerous approaches encountered in the literature, let us mention the classical premium calculation principles based on probabilistic models (see e.g. \cite{Mikosch}, \cite{Asmussen2010}, \cite{Buhlmann1996}),  the axiomatic approach where premium principles are subject to a set of desirable properties (see \cite{Artzner1999}), the  abstract/functional analytic approach where risk measures are derived from an acceptance set and a set of scenario measures (see \cite{Follmer}),   distortion-based measures \cite{wang_1996}, and eventually the approach to risk excess measures induced by hemi-metrics (see  \cite{Faugeras2018}). These numerous approaches have lead to considerable debate on the pros and cons  of the  risk measures available in the literature, see e.g. \cite{academicresponse}. 

From the magnitude-propensity point of view,
the duality of risk is reflected in its measurement, i.e. in the risk measures themselves.
Some risk measures focus more on the propensity aspect of risk, while some others focus more on the magnitude one, and most risk measures mix both.
Let us illustrate this point with the following comparison of three well-known risk measures:
\begin{itemize}
	\item The Value-At-Risk, defined as the left $\alpha$-quantile of $X$, 
	\begin{equation}
	VaR_\alpha(X):=\inf\{x: P(X\le x)\ge \alpha \}, \quad 0<\alpha<1
	\label{VAR-def}
	\end{equation}
	encodes a propensity into a  magnitude, by setting the $VaR_\alpha(X)$ as the utmost-left $\alpha$-quantile. (Note that some authors define Var as the utmost-right quantile). Value at Risk only controls the probability of a loss, it does not capture the size of such a loss if it occurs. 
	\item The opposite extreme is, for $X$ an essentially bounded random variable, the essential supremum, 
	$$
	\rho_\infty(X):=\mbox{ess}\sup X,
	$$
	which quantifies the maximum magnitude of the loss, but gives no information on the probabilities.
	\item The Expected Shortfall, (also called the Conditional Value at Risk, see \cite{UR2000}),
	\begin{equation}
	ES_\alpha(X):=\E [X|X\ge VaR_\alpha(X)],
	\label{ES-def}
	\end{equation}
	clearly  mixes  the two aspects magnitude-propensity of the risk by computing a weighted average over a threshold.
\end{itemize}
In general, one wants to know both when/how often a catastrophe may occur, and also what is the size/extent of the loss one has to face. This suggests that 
reducing risk $X$ to a single proxy on the magnitude scale as a univariate risk measure $\rho(X)$ is somehow inadequate to account for the dual nature of risk.
This idea that one numerical quantity cannot hedge against risk has already been evoked beforehand in the literature, see e.g. \cite{RK}. Let us also remark that in their criticism of the VaR measure,  the authors of the academic response to the Basel 3.5 framework  are implicitly interested in these dual propensity and magnitude aspects of risk (see \cite{academicresponse} p.27  ``Question W1: VaR says nothing concerning the what-if question: Given we encounter a high loss, what can be said about its magnitude?'').  
It thus would be of considerable practical interest to quantify risk on both the magnitude and propensity scales. 
This is the purpose of this paper.

\begin{remark}[On elicitability]
     It has been argued in the literature that elicitability is also a desirable property for risk measures (see \cite{academicresponse}). Elicitability is a concept derived from point forecasting. Roughly speaking, in order that a point forecast, derived from a statistical functional,  be consistent with their evaluation by averaging over the past, the statistical functional must be written as a (strict) M-functional,   see \cite{Gneiting2011} for a precise definition. It is known that VaR is not a coherent risk measure, but is elicitable, while the expected shortfall (ES) is a coherent law invariant risk measure, but is not elicitable.  ES is jointly elicitable with VaR, see \cite{FisslerZiegel}.  
     
     This consideration of elicitability suggests the use of bivariate risk measure (ES,VaR), i.e. of combining magnitude and propensity type univariate risk measures to obtain something that is both coherent and elicitable. This is another supplementary motivation for arguing that one should switch to bivariate risk measures for the study of univariate risks.
\end{remark}
\begin{remark}[On parametrized risk measures]
\label{rem:choice}
For  risk measures depending on a parameter, like $\alpha$ in $VaR_\alpha(X)$ or $ES_\alpha(X)$, arise the practical issue  of choosing the ``right'' parameter value $\alpha$ for correctly representing the risk by a single numerical quantity. Common practice is to hedge against a ``rare event'', i.e. to take, say,  $\alpha=0.9$, $0.95$ or $0.99$. 
 
 $VaR$ and $ES$ give in fact curves, $\alpha\to VaR_\alpha(X)$
 and $\alpha\to ES_\alpha(X)$. These curves, under mild conditions, determine the distribution of $X$. Hence, they give the same information as the c.d.f. $F_X$. They are just another possible analytical characterization of the distribution  of $X$. 
 A complete assessment of the magnitude and propensity effects of a risk $X$ can be visualised by plotting its c.d.f. $F_X$,  or more conveniently its survival function $\overline{F}_X$, possibly on a log-log scale. This is the classical approach in the Engineering literature, see  \cite{kaplan1981quantitative}, where it is argued that ``a single number is not a big 
 enough concept to communicate risk.  It takes a whole curve''.
 
 However, it becomes difficult to compare entire curves. Therefore, it is natural to look for a summary of this  distribution-determining curve to  as few as possible numerical quantities.   This is the path favoured   in Insurance and Financial Mathematics with   risk measures.
   The  magnitude-propensity measure to be introduced below, can thus be viewed as a middle-ground between univariate risk measures of Insurance Mathematics and the full curve approach in Engineering.  
\end{remark}


\subsection{M-statistical functionals can be obtained from  mass transportation to a degenerate distribution}\label{section1-3}

The following discussion gives the key insight for defining a magnitude-propensity risk measure as an optimal transportation problem.
Let us recall that the Monge-Kantorovich optimal transportation problem aims at
finding a joint measure $P^{(X,Y)}$ on, say, the product measurable space $(\mathbb R\times \mathbb R,\mathcal B(\mathbb R^2)$, with prescribed marginals $(P^X, P^Y)$, which is the
solution of the optimisation problem:
\begin{equation}
\mathcal T_c(P^X, P^Y) := \inf_{P^{(X,Y)}\in \mathcal P(P^X,P^Y)} \int c(x, y) P^{(X,Y)}(dx,dy),
\label{KR-def}
\end{equation}
where $c: \mathbb R\times \mathbb R\to \mathbb R^+$ is a cost function and the infimum is on the set $\mathcal P(P^X,P^Y)$ of joint distributions $P^{(X,Y)}$ with given marginals $P^X, P^Y$. Informally, mass at
$x$ of $P^X$ is transported to $y$, according to the conditional distribution $P(dy|x)$
of the optimal transportation plan $P^{(X,Y)}$, in order to recover $P^Y$
while minimising the average cost of transportation $\int c(x, y) P^{(X,Y)}(dx,dy)$. Under regularity conditions, the optimal transportation plan is induced by a (Monge) mapping $T$, viz. $P^{X,Y}=P^{X,T(X)}$. See \cite{rachev1998mass}, \cite{Villani2003}, \cite{villani2008optimal}, \cite{santambrogio2015optimal} for book-length treatment on
the subject.

When $P^Y$ is degenerate, i.e. $P_Y=\delta_m$ with $m\in\mathbb R$, 	 then  $\mathcal P(P^X,\delta_m)$ reduces to the singleton product measure  $\mathcal P(P^X,\delta_m)=\{P^X(dx)\times \delta_m(dy)\}$.  (\ref{KR-def}) thus simplifies as  the expected cost between $X\sim P^X$ and a fixed point $m$,  
\begin{equation}
\mathcal T_c(P^X,\delta_m)=\int c(x,m) P^X(dx) = \E c(X,m) \label{KR-degenerate}.
\end{equation}
Therefore, minimizing the transportation cost (\ref{KR-def}) over the set $\mathcal D:=\{\delta_m(dy), m\in\mathbb R\}$ of Dirac measures is equivalent to minimizing the expected cost (\ref{KR-degenerate}) over $m\in\mathbb{R}$,
\begin{equation}
\mathcal T_c(P^X, \mathcal D)=\inf_{P^{Y}\in \mathcal D} \mathcal T_c(P^X, P^Y)=\inf_{m\in \mathbb R } \E c(X,m).
\label{KR-equiv}
\end{equation}
 In particular, 
 \begin{itemize}
 	\item for  	$c(x,y)=(x-y)^2$, (\ref{KR-def}) is the squared $L_2-$Wasserstein metric $W_2$
 	 	and (\ref{KR-equiv}) writes as the variance, 
 	$$W_2^2(P^X,\mathcal D)=\inf_{P^{Y}\in \mathcal D} W_2^2(P^X,P^Y)=\inf_{m\in \mathbb R } \E (X-m)^2=Var(X),$$
     	and is obviously minimised for the mean $m=EX$. In addition, when $P^X$ is replaced by the conditional law of $X$ given $X\ge VaR_\alpha(X)$, one obtains the Expected shortfall (\ref{ES-def}).
 	\item  For the $L_1$ distance, $c(x,y)=|x-y|$, one gets the median.
 	\item For the asymmetric cost $c(x,y)=(x-y)\alpha \mathds 1_{x-y\ge 0}+(y-x)(1-\alpha) \mathds 1_{y-x> 0}$, with $0<\alpha<1$, 
 	one obtains the (left)-$\alpha$-quantile, $m=q_\alpha(X)$ (see e.g. Koenker and Basset), that is to say the Value-At-Risk (\ref{VAR-def}).
 	
 	\item For  $c(x,y)=y+\frac{(x-y) \mathds1_{x\ge y}}{1-\alpha}$, one gets  simultaneously the Expected Shortfall and the Value-at-risk, when $\E [X]<\infty$: the optimal value in (\ref{KR-equiv}) is the Expected Shortfall while the argmin gives the Value-At-Risk, 
 	see \cite{UR2000}, \cite{rockafellar2002conditional}.
 	
 	\item And so on for other statistical functionals. See also \cite{Faugeras2018} for optimal transportation induced by cost functions which are hemi-metrics encoding an order.
 \end{itemize}   
The above discussion shows that the statistical functionals and risk measures,  which can be expressed as an M-estimator solving (\ref{KR-equiv}) for a suitable cost function, can be regarded as being obtained from a special mass transportation problem towards a family of degenerate Dirac distribution $\delta_m$.

\subsection{The magnitude-propensity $(m_X,p_X)$ approach to measuring risk} 
The optimal transportation view on risk measures of Section \ref{section1-3} suggests that the limitations of the risk measures of  Section \ref{section1-2} come from the fact that these univariate functionals can be viewed as being obtained by mass transportation from the $P^X$ measure to a degenerate Dirac $\delta_m$ measure: the latter distribution only bears a magnitude $m$ with full propensity. Hence, the magnitude and propensity aspects of $P^X$ are mixed and encoded in the sole magnitude $m$ of the Dirac destination measure. 

It therefore becomes natural to suggest  a mass transportation approach to risk measures, as in (\ref{KR-equiv}), but with the target Dirac distribution $\delta_m$ replaced by a two-points distribution $P^Y$,
\begin{equation}
P^Y=(1-p)\delta_0+p\delta_m. \label{def-B}
\end{equation}
The latter distribution encodes both the magnitude and propensity aspects of risk: a loss of magnitude $m$ occurs with probability $p$ (and no loss occurs with probability $1-p$). 
We can therefore  define of the basic idea of the paper as follows:
\begin{definition}
	Let $\mathcal A_0:=\{P^Y=(1-p)\delta_0+p\delta_m,\quad p\in(0,1), m\in\mathbb R^+\}$ be the set of such two-point distributions (\ref{def-B}). For $X\sim P^X$ with $\E[X^2]<\infty$, the bivariate magnitude-propensity risk measure
	$(m_X, p_X)$
	is obtained by minimizing   the Wasserstein $W_2$ distance from  $P^X$ to $\mathcal A_0$,   
	\begin{equation}
	(m_X, p_X)=\arg\inf_{P^Y\in\mathcal A_0} W_2(P^X,P^Y) \label{def}.
	\end{equation}
\end{definition}
In other words, the risk $X$ is approximated in Wasserstein metric by a proxy $Y\sim B_0$, which bears a loss of magnitude $m_X$, occurring with a probability $p_X$, and no loss with probability $1-p_X$. We exclude the values $p=0$, $p=1$ and $m=0$ in the family $\mathcal A_0$, in order to obtain non degenerate measures.

This approach of approximating a distribution by a discrete one corresponds to the well-known problem of optimal quantization, see \cite{GL2000}. The latter is itself related to   $k-$means clustering, see \cite{ Pollard1982a}. Here, the main difference is that one location of the discrete distribution is constrained to be $0$.  This point mass at zero encodes the absence of  loss and thus the point mass $p$ at $m>0$ encodes loss.
$P^Y\in\mathcal{ A}_0$  is characterized by the two parameters $(m,p)$, and quantization to $\mathcal A_0$ gives the minimal way to represent the magnitude and propensity effect borne by $X$.


\section{Theoretical analysis}\label{section2}

The theoretical study of the optimal $(m_X,p_X)$ magnitude-propensity pair in   (\ref{def}) can be done by two main approaches:  either as an optimal mass transportation problem or as an optimal quantization problem. The more refined results are obtained by the latter approach. However, we first make use of the former  to quickly obtain a characterisation of the optimal magnitude-propensity pair.
	
\subsection{Computation of $(m_X,p_X)$ by the optimal transportation approach}
A direct optimization approach can be used to characterize the $(m_X,p_X)$ by using the explicit form of the Wasserstein metric in dimension one: denote by $Q_X$ the quantile function of $P^X$,
$$
Q_X(t):=\inf\{x:F_X\ge t\}, \quad 0<t<1.
$$
It is then   well-known (see e.g. \cite{rachev1998mass}) that  
\begin{equation}
    W_2^2(P^X,P^Y)=\int_0^1 \left(Q_X(t)-Q_Y(t)\right)^2dt,
\label{eq:explicit W2}
\end{equation}
for univariate $P^X, P^Y$ with finite variance.
For $Y  \in\mathcal A_0$, its quantile function writes $Q_Y(t)=m\mathds 1_{1-p<t\le 1}$. Plugging the later in (\ref{eq:explicit W2}) and optimizing gives the following characterization result.

\subsubsection{Characterization by direct optimisation}
In the remainder of the paper, we simplify notations and will denote simply by $F$, $Q$, $f$ the c.d.f., quantile function and density (if it exists) of $X$.
\begin{proposition}
\label{prop}
Let $X$ s.t. $\E[X^2]<\infty$.
	\begin{itemize}
	    \item [i)] Necessary conditions: 
if $(m_X,p_X)$ is a local minimum of  (\ref{def}), then it satisfies the following system of equations,
	\begin{align}
	m_X&=2Q(1-p_X), \label{eq}\\
	m_X&=\frac{\int_{1-p_X}^{1} Q(t)dt}{p_X}.\label{eq2}
	\end{align}
	
    \item [ii)] Sufficiency condition: if $P^X$ is absolutely continuous with continuous density $f$, then $(m_X,p_X)$ is optimal if $f(Q(1-p_X))>0$ and 
    \begin{equation}
    \frac{2p_X}{f(Q(1-p_X))}-Q(1-p_X)>0.
    \label{eq:sufficientcond}
    \end{equation}
    
	\end{itemize}

\end{proposition}

\begin{proof}
	\begin{itemize}
	    \item [i)]
	
	 The squared Wasserstein distance between $P^X$ and $P^Y\in\mathcal A_0$ writes
	\begin{eqnarray*}
	 W_2^2(P^X,P^Y)&=&\int_0^{1-p}(Q(t))^2dt+\int_{1-p}^1(Q(t)-m)^2dt\\
	 &=&E[X^2]+m^2p-2m\int_{1-p}^1Q(t)dt\\
	 &:=&\psi(m,p)
	\end{eqnarray*}
		$\psi$ is differentiable and any optimal magnitude-propensity $(m_X,p_X)$ solving (\ref{def}) must satisfy the first order conditions
		\begin{eqnarray*}
		 \begin{cases}
			\frac{\partial \psi(m_X,p_X)}{\partial m}=0\\
			\frac{\partial \psi(m_X,p_X)}{\partial p}=0
		\end{cases}
	\Leftrightarrow \begin{cases}
		-2\int_{1-p_X}^{1} Q(t)+2m_Xp_X=0\\
		-2m_XQ(1-p_X)+m^2_X=0.
	\end{cases}
		\end{eqnarray*}
	For $m_X\neq 0$ and $p_X\neq 0$, one gets
	 (\ref{eq}) and (\ref{eq2}).
\item [ii)] if $P^X$ has density $f$ such that $f(Q(p))>0$, then $Q$ is differentiable with derivative the quantile-density $q_X(p)=Q(p)'=\frac{1}{f(Q(p))}$. Then $\psi$ is twice differentiable with Hessian matrix
$$\begin{bmatrix}
\frac{\partial^2 \psi}{\partial m^2}&\frac{\partial^2 \psi}{\partial m\partial p}\\
\frac{\partial^2 \psi}{\partial m\partial p}&\frac{\partial^2 \psi}{\partial p^2}
\end{bmatrix}
=\begin{bmatrix} 
a & b\\
b&c
\end{bmatrix}
$$
with $a=2p$, $b=2m-Q(1-p)$, $c=\frac{2m}{f(Q(1-p))}$.
The Hessian is positive-definite at the critical point $(m_X,p_X)$ if $a>0$ and $ac-b^2>0$.
The latter condition writes 
$$
ac-b^2=\frac{4p_Xm_X}{f(Q(1-p_X))}-4\left(m-Q(1-p_X)\right)^2>0
$$
With $(m_X,p_X)$ satisfying (\ref{eq}), (\ref{eq2}), the condition writes
$$
4Q(1-p_X)\left( \frac{2p_X}{f(Q(1-p_X))}-Q(1-p_X)\right)>0,
$$
which is (\ref{eq:sufficientcond}).

\end{itemize}
	
\end{proof}

\begin{remark}
\begin{enumerate}
    \item If $1-p_X$ is in the range of $F$ (in particular if $F$ is continuous), then,   (\ref{eq}) writes 
    $$F(m_X/2) =1-p_X,$$ and by Exercise 3.3 in \cite{Shorack2000} p. 113,  (\ref{eq2}) writes as \begin{equation*}
    m_X=\frac{\E[X \mathds 1_{X> m_X/2}]}{1-F(m_X/2)}=\E[X| X> m_X/2].
    \end{equation*}
    This equation will be derived in the general case    from the optimal quantization viewpoint, see Theorem \ref{thm:existence} and Remark \ref{rem:general} below.
    
    \item Setting $a:=Q(1-p_X)$, the sufficiency condition (\ref{eq:sufficientcond}) writes also
    $$af(a)<2(1-F(a)).$$
    In view of the fact that
    $\E  [X]=\int_0^\infty tf(t)=\int_0^\infty (1-F(t))dt<\infty$, the latter condition has to be understood as a condition on the tail decrease of the density. It can also be expressed as a condition on the failure rate (or hazard function), as
    $$h(a):=\frac{f(a)}{1-F(a)}<\frac{2}{a}.$$
    A general existence   result will be given in Theorem \ref{thm:existence}, by quantization methods.

\end{enumerate}
    
\end{remark}

\subsubsection{Examples}
We illustrate the formulas obtained in   Proposition \ref{prop} on the following examples:
\begin{example}[Uniform distribution]
\label{ex:uniform}
	For $X\sim U_{[0,a]}$, i.e. $Q(p)=ap$, (\ref{eq}) and (\ref{eq2}) give  $p_X=2/3$ and $m_X=2a/3$.
\end{example}

\begin{example}[Exponential distribution]
\label{ex:exponential}	For $X\sim \mbox{Exp}(\lambda)$, with $\lambda >0$, the memorylessness property yields that $\E [X|X>a]=a+\lambda^{-1}$. Hence, the optimal threshold is $a_X=\lambda^{-1}=\E [X]$ and  one has  $p_X=e^{-2}\approx 0.135$ and $m_X=\frac{2}{\lambda}=2 \E[X]$.
\end{example}
It is noteworthy that, in both examples,  $p_X$ does is fixed and does not depend on the parameter ($a$, resp. $\lambda$) of the distribution.

\begin{example}[Pareto distribution]
\label{ex:pareto}
    One considers the following one-parameter (version) of the Pareto Distribution, $X\sim Pa(\theta)$, defined by $P(X>x)=(1+x)^{-\theta}$, for $x\ge 0$. We assume that $\theta>2$, so that $\E[ X^2]<\infty$.  One has that 
    $$\E [X|X>a]=\frac{\theta}{\theta-1}(1+a)-1.
    $$
    Using the characterisation (\ref{eq:condition theorem}), one finds that the optimal threshold is
    $a=\frac{1}{\theta-2}
    $.
    Hence, (\ref{eq:sol theorem}) gives  
    $$m_X=\frac{2}{\theta-2},\quad  p_X=\left(\frac{\theta-2}{\theta-1}\right)^\theta.$$
\end{example}

For some other distributions, one may not have a closed form expression.

\subsection{A review of Optimal Quantization}
\label{sec:quantization}
Additional insight is gained by viewing the basic approach (\ref{def}) as a constrained optimal quantization problem. We first summarize the basic facts and terminology of (unconstrained) optimal quantization theory. The following  can be regarded as a quick introduction to the field.

Optimal  quantization originates from the engineering and signal processing literature (see \cite{Lloyd}, \cite{gersho1992vector}). It aims at  optimally discretizing  a continuous (stationary) signal in view of its transmission. It was developed originally for analog-to-digital conversion,  compression, pattern recognition.

Following \cite{gersho1992vector} and \cite{GL2000}, a  $N-$vector quantizer on $(\mathbb R^d,||.||)$   is a mapping $T:\mathbb R^d\mapsto \{x_1,\ldots,x_N\}$, where $\{x_1,\ldots, x_N\}$ is a codebook of size $N$. Thus, associated with $T$ is a partition $\{A_i\}$, with $A_i=\{x\in\mathbb R^d: T(x)=x_i\}$, of the input space so that a quantizer writes
$$T(x)=\sum_{1=1}^N x_i \mathds 1_{A_i}(x), \quad x\in \mathbb R^d$$
and is determined by the pairs $\{(x_i,A_i), i=1,\ldots n\}$.
An $N-$optimal quantizer for a distribution $P^X$ is a $N-$quantizer
which minimises the mean squared error (or distortion)
$$
D(T; P^X):=\inf_{T} \E(X-T(X))^2.
$$
Equivalently, it can be shown (\cite{Pollard1982a}, \cite{GL2000}) that an optimal quantizer is a Monge map  minimising the Wasserstein metric $W_2(P^X, P^Y)$ between $P^X$ and $P^{Y}$, where $P^Y$ is a discrete measure with $N$ points.

It is known, see \cite{gersho1992vector}, that given the codebook, the optimal partition is given by the Voronoi partition,
$$
A_i:=\{x\in\mathbb R^d: ||x-x_i||\le ||x-x_j||, j\neq i \}.
$$
Conversely, given a partition, the optimal codebook is given by the centers (centroids),
$$
y_i=\E(X|X\in A_i).
$$
These properties are the basis of Lloyd's iterative algorithm. This  also explains why optimal quantizers are sometimes defined directly by their codebook and corresponding Voronoi partition, see e.g. \cite{Pages2018}. As a consequence, the distortion writes as a sole function of the centers,   as 
\begin{equation}
D(T; P^X)=D_N(x_1,\ldots, x_N):= \E \min_{x_i} ||X-x_i||^2
\label{eq:distortion centers}
\end{equation}

If the support of $P^X$ has at least $N$ elements, existence of an $N$-optimal quantizer follows from the fact that $(x_1,\ldots,x_N)\to \sqrt{D_N(x_1,\ldots, x_N)}$ is 1-Lipshitz and a non-trivial compacity argument, see e.g. Lemma 8 in \cite{Pollard1982a}, Theorem 5.1 in \cite{Pages2018} or Theorem 4.12 in \cite{GL2000}. In dimension one, a known sufficient condition for uniqueness of the $N$-optimal quantizer is that $P^X$ be absolutely continuous with a log-concave density $f$. By Proposition 6.6 in \cite{Pages2018}, if $x^N:=(x_1,\ldots, x_N)$ has pairwise distinct component and $P^X\left( \cup_{1\le i\le N} \partial A_i\right)=0$, then $D_N$ is continuously differentiable with gradient 
$\nabla D_N=\begin{bmatrix}\frac{\partial D_N}{\partial x_i}\end{bmatrix}$ with \begin{equation}
    \frac{\partial D_N}{\partial x_i} (x^N)=\E[2(x_i-X)\mathds 1_{X\in A_i}].
    \label{eq:derivative}
\end{equation} 

These calculations are the basis for a stochastic gradient descent based algorithm known as Competitive Learning Vector Quantization, see   \cite{Pages2018}.  In the one-dimensional case,  explicit expressions
of the gradient can be written in terms of the c.d.f. $F$ and of the cumulative first moment function $K(x):=\E[X\mathds 1_{X\le x}]$ as
$$
\frac{\partial D_N}{\partial x_i}(x^N) =2x_i\left[F(x_{i+\frac{1}{2}}) -F(x_{i-\frac{1}{2}})\right] -2\left[K(x_{i+\frac{1}{2}})-K(x_{i-\frac{1}{2}}) \right],
$$ 
where $x_{i+\frac{1}{2}}=\frac{x_{i+1}+x_i}{2}$ corresponds to the boundaries of the Voronoi cells.
Additional formulas for the Hessian are available (see \cite{Pages2018} p. 155). These   formulas allow for a Newton-Raphson zero search procedure.

\subsection{Existence and characterization of  $(m_X,p_X)$ by constrained optimal quantization}
We can now tackle the study of $(m_X,p_X)$ from the optimal quantization viewpoint.
One thus  introduces the constrained two-points quantizer with centers $\{x_0,x_1\}:=\{0,m\}$, i.e. as a mapping $T:\mathbb R^+\mapsto \{0,m\}$ with $T(x)=m \mathds 1_{x\ge a}$, where $a$ is a threshold to determinate. Then, the optimal quantization problem with constrained knot at zero writes
$$
\inf_{a,m\in\mathbb R^+} \E[(X-T(X))^2].
$$
By the results of the previous section (see (\ref{eq:distortion centers})), one already knows that  $a=m/2$, i.e. that the Voronoi regions write $A_0=\{x: 0\le x\le m/2\}$ and $A_1=\{x: x\ge m/2\}$.
As a consequence, and in view of (\ref{eq:distortion centers}) the distortion/objective function writes as a sole function of the magnitude $m$, as
\begin{equation}
    L(m):=\E[X^2\wedge (X-m)^2].
    \label{eq:objective}
\end{equation}
We first study the differentiability properties of the objective function (\ref{eq:objective}) in the next Section.

\subsubsection{B-differentiability properties of the objective function}
 For the general $N$-points quantization problem,  the square root of the distortion function is $1-$Lipshitz, see \cite{Pages2018} p. 136. Hence, it has a derivative a.e. by Rademacher's Theorem. Therefore, $L$ is differentiable a.e. 
In addition, since the integrand $H:x\to X^2\wedge (X-x)^2$ is piecewise differentiable, it is easy to show that it has directional derivatives everywhere. A convenient tool in the setting of optimisation of piecewise smooth functions is the concept of Bouligand derivative (B-derivative), see  \cite{Robinson1987}  and \cite{Scholtes}. It drops the requirement of linearity of the differential, represents a first-order approximation and allows to have a single-valued notion of differential. We give below a simplified definition for functions of one variable.
 \begin{definition}
A function  $f:\mathbb R \to \mathbb R$ is Bouligand differentiable (B-differentiable) at $x_0$ if there exists a positive homogeneous function $\nabla^B f(x_0):\mathbb R\to \mathbb R$ s.t.
\begin{equation}
 f(x_0+v)=f(x_0) +\nabla^B f(x_0)(v) +o(v),\quad \forall v\in\mathbb R.
 \label{eq:def B-derivative}   
\end{equation}
 \end{definition}

For the study of the B-differentiability of the objective function $L$, we first give a lemma on the differentiability of its integrand.
\begin{lemme}[B-derivative of a min function]
\label{lemme:B-derivative min}
Let $H(x)=\min(h_1(x),h_2(x))$, where $h_1,h_2:\mathbb R\to \mathbb R$ are differentiable functions.
Then $H$ is B-differentiable, with B-derivative given by
\begin{enumerate}
    \item [i)] if $x$ is s.t. $h_1(x)<h_2(x)$,  then $\nabla^B H(x)(v)=h_1'(x)v$, and conversely, if $h_1(x)>h_2(x)$, then $\nabla^B H(x)(v)=h_2'(x)v$.
    \item [ii)] if $x$ is s.t. $H(x)=h_1(x)=h_2(x)$, then $\nabla^B H(x)(v)= \min(h_1'(x)v,h_2'(x)v)$.
\end{enumerate}\end{lemme}

\begin{proof}
\begin{itemize}
    \item [i)] Assume w.l.o.g. that $h_1(x)<h_2(x)$, so that $H(x)=h_1(x)$. Let us show that $H(x)=h_1(x)$ remains true on a neighborhood of $x$. Set $d=h_2(x)-h_1(x)>0$. Let $0<\varepsilon_1<d$, $0<\varepsilon_2<d-\varepsilon_1$. By continuity of $h_1$ and $h_2$ at $x$, there exists $\delta>0$ s.t. $\forall |v|<\delta$,
    \begin{align}
        h_1(x+v)&\le h_1(x)+\varepsilon_1\label{1}\\
        h_2(x+v)&\ge h_2(x)-\varepsilon_2\label{2}
    \end{align}
    $h_1(x)+\varepsilon_1=h_2(x)+\varepsilon_1-d$, therefore (\ref{1}) and (\ref{2}) give
    $$
    h_1(x+v)\le h_2(x)+\varepsilon_1-d< h_2(x)-\varepsilon_2\le h_2(x+v).
    $$
    Hence, $H(x+v)=h_1(x+v)$ for all $|v|< \delta$. Thus, the Bouligand derivative of $H$ at $x$ reduces to  the  classical derivative $h_1'(x)$ of $h_1$. 
    \item [ii)] One has
    \begin{align}
        h_1(x+v)&=h_1(x)+h_1'(x)v +o(v)\nonumber\\
        h_2(x+v)&=h_2(x)+h_2'(x)v+o(v)\label{eq:Taylor}
    \end{align}
    Since $H(x)=h_1(x)=h_2(x)$, 
    $$
    H(x+v)-H(x)=\min(h_1(x+v)-h_1(x),h_2(x+v)-h_2(x)).
    $$
    We then use the inequality 
    $$
    |\min (a,b)-\min (c,d)|\le\max (|a-c|,|b-d|),
    $$
    applied to $a=h_1(x+v)-h_1(x)$, $b=h_2(x+v)-h_2(x)$, $c=h_1'(x)v$, $d=h_2'(x)v$, to deduce that
    $$
    |H(x+v)-H(x)-\min(h_1'(x)v,h_2'(x)v)|\le \max(|o(v)|,|o(v)|)=|v| \max(|o(1)|,|o(1)|),
    $$
    viz. $H$ is B-differentiable at $x$, with B-derivative
    $$
    \nabla^BH(x)(v)=\min(h_1'(x)v,h_2'(x)v),
    $$
    as the latter expression is positively homogeneous.
    \end{itemize}
\end{proof}

\begin{corollary}[B-Differentiability of the objective function]
\label{cor:B-derivative objective}
$L$ is B-differentiable on $x\ge 0$, with B-derivative given by
$$
\nabla^B L(0)(v)=-2\E [X]v\mathds 1_{v>0},
$$
and, for $x>0$, 
\begin{equation}
\nabla^B L(x)(v)=2v\E[(x-X)\mathds 1_{X>x/2}] +P(X=x/2)\min (0,xv).
\label{eq:B derivative objective}
\end{equation}
\end{corollary}
\begin{proof}
Applying Lemma \ref{lemme:B-derivative min} to the integrand function $H(x)=X^2\wedge (X-x)^2$, i.e. with $h_1(x)= X$, $h_2(x)=  (X-x)^2$, gives  the B-derivative of the integrand function: almost surely, for $v\in\mathbb R$,
\begin{equation*}
    \nabla^BH(x)(v)=\begin{cases}
    0& x<0\\
    \min(0,-2Xv) &x=0\\
    2(x-X)v& 0<x<2X\\
    \min(0,2Xv) &x=2X\\
    0& x>2X
    \end{cases}
\end{equation*}
In addition, the  second order terms in the Taylor expansions (\ref{eq:Taylor})   of $h_1$ and $h_2$ (which are exact) do not depend on $X$, so that reasoning as in \ref{lemme:B-derivative min}, it is easy to see that one gets a remainder term in (\ref{eq:def B-derivative}) for $H$ which does not depend on $X$, with at most linear growth:
$$
H(x+v)=H(x)+\nabla^BH(x)(v)+v\epsilon(v),  
$$
a.s., with $|\epsilon(v)|\le |v|$.
Therefore, if one sets 
$$
\nabla^B L(x)(v):=\E[\nabla^BH(x)(v)],
$$
then, by linearity,
\begin{align*}
    \frac{|L(x+v)-L(x)-\nabla^B L(x)(v)|}{|v|}\le |v|\to  0,
\end{align*}
as $|v|\to 0$.
 By integration,
$$
\nabla^B L(0)(v)=\begin{cases}
-2\E [X]v& v>0\\
0& v\le 0
\end{cases}
$$
 and for $x>0$,
\begin{align*}
\nabla^B L(x)(v)&=\int \left(2(x-t)\mathds 1_{0<x<2t}+\min(0,2tv)\mathds 1_{x=2t}\right) P^X(dt)\\
&=\E[2(x-X)\mathds 1_{X>x/2}]v+P(X=x/2)\min (0,xv).
\end{align*}
 $v\to \nabla^B L(x)(v)$ is positively homogeneous, for $x\ge 0$. Hence, $L$ is B-differentiable.

\end{proof}

\subsubsection{Main existence and characterization result}
With these tools, we can now give the main existence and characterization result of the optimal magnitude-propensity pair.
\begin{theorem}
\label{thm:existence}
\begin{enumerate}
    \item [i)] If   $\E [ X^2]<\infty$ and the support of $P^X$ contains at least two points, then there exists a magnitude-propensity pair $(p_X,m_X)$ minimizing (\ref{def}).
    \item [ii)] An optimal magnitude-propensity pair $(p_X,m_X)$ is characterized as a solution of the equation
    \begin{equation}
    2a=\E[X|X> a],\quad a>0,
    \label{eq:condition theorem}
    \end{equation}
    and then 
    \begin{equation}
        m_X=2a, \quad p_X=P(X>m/2).
        \label{eq:sol theorem}
    \end{equation}
    In addition, $P^X(\{a\})=0$.
\end{enumerate}
\end{theorem}

\begin{proof}
    \begin{enumerate}
        \item [i)] 
             By the results of the Section \ref{sec:quantization}, the objective function $L$ of (\ref{eq:objective}) is locally Lipshitz hence continuous and its lower level sets are compact, see e.g. Lemma 8 in  \cite{Pollard1982a}. Hence,  the general existence result follows from Weierstrass theorem. See also  Exercise 3 p. 139 in \cite{Pages2018}.
        \item [ii)] 

        Since $\nabla^B L(.)(.)$ gives a first order approximation of $L$, a necessary condition for $m_X$ to be a minimizer is that $\nabla^B L(m_X)(v)\ge 0$, $\forall v\in\mathbb R$. For $m_X=0 $, $\nabla^B L(0)(v)=-2\E[X]<0$ for $v>0$, therefore $0$ cannot be a critical point. For $m_X>0$ and $v>0$, the first order condition writes 
        \begin{equation}
            \E[(m_X-X)\mathds 1_{X>m_X/2}]=0.
            \label{eq:first1}
        \end{equation}
        For $m_X>0$, $v<0$, it reduces to 
        \begin{equation}
            2\E[(m_X-X)\mathds 1_{X>m_X/2}]+m_XP(X=m_X/2)=0.
            \label{eq:first2}
        \end{equation}
        Plugging (\ref{eq:first1}) into (\ref{eq:first2}) implies that $P(X=m_X/2)=0$ and yields  (\ref{eq:condition theorem}), upon reparametrizing by $a=m_X/2$.
    \end{enumerate}
\end{proof}

\begin{remark}
\label{rem:general}
    \item [i)] We hopefully obtain the same formula   between (\ref{eq:condition theorem}) of  Theorem \ref{thm:existence} and (\ref{eq}), (\ref{eq2}) of Proposition \ref{prop}. 
    The quantization approach allows to obtain a general existence  result.  
    \item [ii)]  $P^X$ does not charge the   optimal threshold $a=m_X/2$, even if $P^X$ is not absolutely continuous. By (\ref{eq:B derivative objective}), this imply that $\nabla^B L(m_X)(v) $ is linear in $v$, that is to say,   $L$ is differentiable (in the classical sense) at the optimal $m_X$, with
    $$L'(m_X)=2\E[(m_X-X)\mathds 1_{X>m_X/2}]=2\E[(m_X-X)\mathds 1_{X\ge m_X/2}].
    $$
\end{remark}

\subsubsection{Uniqueness}
The matter of uniqueness is notoriously a difficult topic in optimal quantization, see \cite{GL2000}, \cite{Pages2018}. 
We provide below in Theorem \ref{thm:log concave full support} a set of sufficiency conditions  for the uniqueness of the magnitude-propensity pair. The main condition is the log-concavity of $x\to x^3f(x)$, which is fulfilled when the density $f$ is itself log-concave. It is a condition similar to  the classical condition   of \cite{Trushkin82} and \cite{Kieffer1983}   for the uniqueness in (unconstrained) optimal quantization. 

\begin{theorem}
\label{thm:log concave full support}
Let $P^X$ be an absolutely continuous distribution on $\mathbb R_+$,  with  density $f$, $0\in\texttt{supp}(P^X)$ and $\E[X^2]<\infty$. Assume 
\begin{enumerate}
    \item [i)] $\tilde f:x\to x^3f(x)$ is strictly log-concave on $\mathbb R_+$;
    \item [ii)]   $\lim_{x\to 0} xf(x)=0$. 
\end{enumerate}  
Then,  the optimal magnitude-propensity pair is unique.  
\end{theorem}


\begin{proof}
Prior to the proof, let us make the following elementary remarks. Being log-concave, the function $\tilde f:x\to x^3f(x)$ has a limit in $[0,\infty)$ as $x\to 0$. Thus, we may set $\tilde f(0)=\lim_{x\to 0} x^3f(x)=0$, by assumption ii). Consequently, $I:=\{x>0:\tilde f(x)>0\}$ is an interval of $\mathbb R_+$ of the form $(0,b]$ or $(0,b)$, with $b\in (0,\infty]$, since $0\in\texttt{supp}(P^X)$, and $f(x)=\tilde f(x)/x^3$ is continuous on $(0,b)$. Moreover, still by log-concavity, $\tilde f$ has a left-handed limit $\tilde f(b-)$ at $b$. Finally, $f$ is continuous on $(0,b)$, with left-handed limit $f(b-)\ge 0$.

The optimal magnitude $m_X$ is s.t. $y_*=m_X/2$ is a zero of  the function 
\begin{equation}
    \Phi(y)=2y \overline{F}(y)-\overline{K}(y),
    \label{eq:Phi}
\end{equation}
where we have set $\overline{F}(y)=\int_y^b f(t)dt$ and $\overline{K}(y)=\int_y^b tf(t)dt$.  
Moreover,  since $\tilde f$ is log-concave, $\ln f$ has right and left-handed derivatives on $(0,b)$, and, by monotonicity, these one-sided derivatives have limits as $x\downarrow 0$ and $x\uparrow b$. 
Therefore, $\Phi'$   can be continuously defined on $y\in[0,b]\cap \mathbb R^+$, as
\begin{equation}
    \Phi'(y)=2\overline{F}(y)-yf(y) 
    \label{eq:Phi'}. 
\end{equation}
On $(0,b)$, $\Phi'$ has right and left-handed derivatives, so let us  define $\Phi''_r$ as the right-handed derivative  
\begin{align}
      \Phi''_r(y)&:=-3f(y)-yf'_r(y)=-yf(y)\left(\frac{3}{y}+\frac{f'_r}{f}(y)\right)\nonumber\\
    \ &=-yf(y) \zeta(y),
    \label{eq:Phi''}
\end{align}
where  
\begin{equation}
    \zeta(y):=3/y+f'_r(y)/f(y)
    \label{eq:phi}
\end{equation}
with $f'_r$ denoting the right-handed derivative of $f$.
  By assumption, $x\to x^3f(x)$ is strictly log-concave,  therefore $\zeta$ is decreasing. 

Let us show that there exists a unique $y_1\in (0,b)$ s.t. $\Phi'(y_1)=0$. For that purpose, consider the following cases:
\begin{enumerate}
    \item [Case i)] Assume that $\lim_{y\to b} \zeta(y)=\ell\ge 0$. Then, by strict log-concavity, $\zeta(y)>\ell\ge 0$ on $I$, i.e.  
       $\Phi''<0$ and  $\Phi'$ decreasing on $[0,b]$. By assumption ii), $\Phi'(0)=2\overline{F}(0)-\lim_{y\to 0} yf(y)=2>0$.
    \begin{enumerate}
        \item [Case ia)] if $b=\infty$, then $\lim_{y\to \infty}\Phi'(y)=-\lim_{y\to \infty} yf(y)=0$ since $\E[X]<\infty$. Thus $\Phi'>0$  and $\Phi$ is increasing on $(0,\infty)$. But since $\Phi(0)=-\E[X]<0$ and $\lim_{y\to\infty} \Phi(y)=0$, $\Phi$ cannot have a zero in $(0,\infty)$, which contradicts the general existence Theorem \ref{thm:existence}. Thus, necessarily, $b<\infty$.
        \item [Case ib)] if $b<\infty$, then $\lim_{y\to b}\Phi'(y)=-bf(b-)\le 0$. 
        If $f(b-)=0$, then $\Phi'>0$ on $(0,b)$, so that $\Phi$ is increasing on $[0,b)$ from $-\E[X]$ up to $0$, which contradicts the existence of $y_*=m_X/2\in [0,b/2)$ such that $\Phi(y*)=0$.
        Therefore, $f(b-)>0$, $\lim_{y\to b}\Phi'(y)< 0$ and  there exists a unique $y_1\in (0,b)$ s.t. $\Phi'(y_1)=0$.
    \end{enumerate}
    \item [Case ii)] Assume that $\lim_{y\to b} \zeta(y)=\ell< 0$.  
    \begin{enumerate}
        \item [Case iia)] If $\lim_{y\to 0}\zeta(y)=\eta\le 0$.  Then, since $\zeta$ is decreasing, $\zeta<0$ on $(0,b)$, hence $\Phi''_r>0$, so that  $\Phi'$ is  increasing on $(0,b)$. But this is clearly impossible, since $\Phi'(0)=2$ and $\lim_{y\to b}\Phi'(y)=0$, resp.  $=-bf(b-)\le 0$, for $b=\infty$, resp. $b<\infty$. Thus, necessarily  $\lim_{y\to 0}\zeta(y)=\eta>0$.
        \item [Case iib)] If $\lim_{y\to 0}\zeta(y)=\eta> 0$. Then, since $\zeta$ is decreasing,  there exists a unique  $y_0\in I$ s.t. $\Phi''_r<0$ on $(0,y_0)$, and $\Phi''_r>0$ on $(y_0,b)$.

    Assume that $\Phi'(y_0)\ge 0$. 
    Then, $\Phi'>0$ on $(0,b)\setminus \{y_0\}$ and $\Phi$ is increasing on $[0,b]$. By the general existence Theorem \ref{thm:existence}, there exists some $y^*=m_X/2\in(0,b)$ s.t. $\Phi(y^*)=0$. Yet, $\lim_{y\to\infty} \Phi(y)=0$ since $\E[X]<\infty$, which is a contradiction. Therefore,  $\Phi'(y_0)< 0$. With the fact that $\Phi'(0)=2>0$ and $\Phi'$ is decreasing on $(0,y_0)$, this implies that  there exists a unique $y_1\in(0,y_0)$ s.t. $\Phi'(y_1)=0$. On $(y_0,b)$,  $\Phi'$ cannot have any zeros, since $\Phi'$ is increasing, $\Phi'(y_0)< 0$ and either  $\Phi'(b-)=-bf(b-)<0$ if $b<\infty$, or $\lim_{y\to \infty}\Phi'(y)=0$ if $b=\infty$, since $\E[ X]<\infty$.   
     \end{enumerate}
\end{enumerate}
So, in both cases, $\Phi$ is increasing on $(0,y_1)$ and decreasing on $(y_1,b)$. Moreover $\Phi(0)=-\E[X]<0$ and $\lim_{y\to b}\Phi(y)=0$, and the general existence Theorem \ref{thm:existence} ensures that $\Phi$ has at least a zero in $(0,b/2)$. Hence, necessarily $\Phi(y_1)\ge 0$ and $\Phi$ has at most one zero, located in $(0,y_1)$.
\end{proof}

In particular, if $f$ is itself log-concave, then the conditions of Theorem \ref{thm:log concave full support} are fulfilled, as shown in the next corollary. 
\begin{corollary}
\label{cor:log concave}
If $f$ is a   log-concave density on $(0,b)$, then the assumptions i) ii)  of Theorem \ref{thm:log concave full support} hold.
\end{corollary}

\begin{proof}
\begin{enumerate}
    \item [i)] $x\to 3/x$ is decreasing and $f'_r/f$ is non-increasing, therefore $\zeta$ is decreasing, i.e. $\tilde f:x\to x^3f(x)$ is strictly log-concave.
    \item [ii)] Let $y_0>0$ s.t. $f(y_0)>0$. For $0<y<y_0<b$, one has, by log-concavity, that
    \begin{align*}
    &\ln f\left(\frac{y+y_0}{2} \right)\ge \frac{1}{2}\ln f(y)+\frac{1}{2}\ln f(y_0)\\
    \Longleftrightarrow\quad&    \ln f(y)\le 2\ln f\left(\frac{y+y_0}{2} \right)-\ln f(y_0)\\
    \Longleftrightarrow\quad&    0\le f(y)\le f^2\left(\frac{y+y_0}{2} \right)/ f(y_0)\\
    \Longleftrightarrow\quad&    0\le yf(y)\le yf^2\left(\frac{y+y_0}{2} \right)/ f(y_0),
    \end{align*}
    which implies $\lim_{y\to 0} yf(y)=0$ by continuity. (Note that the statement is trivial if $\lim_{y\to 0} f(y)<\infty$.) 
\end{enumerate}

\end{proof}

\begin{remark}
  Examples of distributions  satisfying the hypotheses of Theorem  \ref{thm:log concave full support}, but not those of Corollary \ref{cor:log concave} (in particular, whose density is not log-concave) are the power distributions
  $$f(x)=(a+1)x^a\mathds 1_{(0,1)}(x),\quad -1<a<0, $$
  and the Gamma distributions
  $$
  f(x)=\frac{x^{a-1}}{\Gamma(a)}e^{-x}\mathds 1_{x>0},\quad 0<a<1.
  $$
\end{remark}

\begin{remark}
If, in addition to being log-concave, the density is non-increasing, a simpler proof of Theorem \ref{thm:log concave full support} can be given as follows:
Let $\tilde X$ be the symmetrized version of $X$, with density $\tilde f(x)=\frac{1}{2}f(|x|)$.
Let $\lambda\in[0,1]$, $x,y\in\mathbb R$. Since $\ln f$ is non-increasing and concave,
$$
\ln f(|\lambda x+(1-\lambda)y|)\ge \ln f(\lambda|x|+(1-\lambda)|y|)\ge \lambda \ln f(|x|)+(1-\lambda)f(|y|),
$$
viz. $\tilde f$ is log-concave. Therefore, by the existence and uniqueness results of \cite{Trushkin82}, \cite{Kieffer1983}, there exists a unique three-points stationary quantizer of $\tilde X$. 

 The optimal magnitude $m_X$ is obtained by deriving the optimal two-points quantizer of $X$ with a point constrained to be zero. $m_X$ must verify the stationary condition 
 \begin{equation}
 \E[(m_X-X)\mathds 1_{X\ge m_X/2}]=0.
 \label{statio m}
 \end{equation}
 On the other hand, a three-points (unconstrained) quantizer $(x_1,x_2,x_3)$ for $\tilde X$ must satisfy the stationary condition $\E[(x-\tilde X)\mathds 1_{A_i}(\tilde X]=0$.  
 Let us show that that $(-m_X,0, m_X)$ is a three-points stationary quantizer for $\tilde X$. 
Indeed, for  $x_1=-m_X$, the stationary condition writes
$\int (-m_X-t)\mathds 1_{ t\le m_X/2}\tilde f(t)dt=0$, which gives (\ref{statio m}) by the change of variable $x=-t$ and the symmetry of $\tilde f$. For $x_3=m_X$, the stationary condition is also (\ref{statio m}) since $m_X>0$. For $x_2=0$, it writes $\int_{-m_X/2}^{m_X/2} x\tilde f(x)dx=0$, which is automatically satisfied since $\tilde f$ is even.
 Therefore, $(-m_X,0, m_X)$ is the unique three points stationary quantizer for $\tilde X$. Thus $m_X$ is unique. 

\end{remark}

\subsection{Discussion and  properties}\label{sec:examples}

The optimal magnitude-propensity pair $(m_X,p_X)$ obtained by (\ref{eq:sol theorem}) or (\ref{eq}), (\ref{eq2})  has several interesting characteristics. From (\ref{eq}), resp. (\ref{eq2}), it is clear that the magnitude $m_X$ can be interpreted either as (twice) a Value-At-Risk $Var_\alpha(X)$, resp.   as an Expected Shortfall $ES_\alpha(X)$, for a special value of $\alpha$:
\begin{equation}
    m_X=ES_{p_X}(X)=2VaR_{1-p_X}(X) .
\end{equation}
It is thus comforting that one obtains, for the magnitude effect, a quantity akin to the classical and well-used univariate risk measures. In addition,  such an interpretation gives an answer to the problem mentioned in Remark \ref{rem:choice} about the ``right'' choice of $\alpha$ for parametrized risk measures like $VaR_\alpha(X)$ or  $ES_\alpha(X)$: the corresponding $\alpha$ is determined by $p_X$, hence by the distribution $P^X$. It is no longer a subjective choice dependent on the user.

The magnitude-propensity paradigm $m_X,p_X$, by quantifying risk on a bivariate scale, is fundamentally different from the traditional approach to measuring risk by a univariate coherent risk measure $\rho$. Hence, it may not make sense to look for properties similar to those of coherent risk measures. However, one has the following   noteworthy properties for the magnitude $m_X$.
\begin{proposition}[positive homogeneity/scaling of the magnitude]
For $a>0$, $m_{aX}=am_X$, and $p_{aX}=p_X$
\end{proposition}
\begin{proof}
    From the expression (\ref{eq:distortion centers}) of the distortion as a function of the centers, 
    $$
    W_2^2(P^{aX},P^Y)=\E[\min\{aX^2,(aX-m)^2\}]=a^2 \E[\min\{X^2,(X-m/a)^2\}].
    $$
    Therefore $\frac{m_{aX}}{a}=m_X$, and $p_{aX}=p_X$.
\end{proof}
Recall the definition of the convex order:  $X\le_{cx} Y$ if $\E \phi(X)\le \E\phi(Y)$, for all integrable, convex functions $\phi$, see e.g. \cite{r2013}. One  has:
\begin{proposition}[Monotonicity of the magnitude]
If $X\le_{cx} Y$, then $m_X\le m_Y$.
\end{proposition}
\begin{proof}
    By e.g. Proposition 1 in \cite{Puccetti2013}, 
    $X\le_{cx} Y$ 
    if and only if $\E [X|X>a]\le \E[Y|Y>a]$, for all $a\ge 0$. Hence, the curve 
    $a\mapsto \E[X|X>a]$ is below the curve $a\mapsto \E[Y|Y>a]$. Thus, the optimal threshold  $a_X$ for $X$, satisfying (\ref{eq:condition theorem}), is lower than the corresponding optimal threshold $a_Y$ for $Y$. Therefore, $m_X\le m_Y$. 
\end{proof}

Eventually, the magnitude effect will be larger than the average of $X$, which seems a desirable property from the point of view of the classical premium calculation principles \cite{Buhlmann1996}:
\begin{proposition}\label{prop:EX}
If $P^X$ is absolutely continuous, then $m_X\ge EX$.
\end{proposition}

\begin{proof}
Let $\overline{F}(x)=P(X\ge x)$, $\overline{K}(x)=\E [X\mathds 1_{X\ge x}]=\int_{(x,\infty)}t f(t)dt$ and $\beta(x)=\overline{K}(x)/\overline{F}(x)$. Then, 
$$
\beta'(x)=\frac{-xf(x)\overline{F}(x)+f(x)\overline{K}(x)}{\overline{F}^2(x)}=\frac{f(x)}{\overline{F(x)}}\left(
\frac{\overline{K}(x)}{\overline{F}(x)}-x\right)\ge 0
$$
as $\overline{K}(x)\ge x\overline{F}(x)$. Therefore, $\beta$ is non-decreasing, hence $\beta(x)\ge \beta(0)=\E[ X]$. Therefore, the solution $m_X$ of the equation $x=\beta(x/2)$ satisfy $m_X\ge \E[X]$.
\end{proof}


	




\section{Magnitude-propensity risk comparison: numerical illustrations and empirical aspects }
\label{sec:section3}
\subsection{Risk comparison with magnitude-propensity plots}
A fundamental objective of risk analysis is the comparison of risks. The proposed approach allows to compare risks on both the magnitude and propensity scales,  by displaying a point of coordinates $(m_X,p_X)$ on the  magnitude and propensity axes. Figure  \ref{fig:mp ex1} illustrates the resulting magnitude-propensity plot, for the uniform, exponential and Pareto distributions of Examples \ref{ex:uniform}-\ref{ex:pareto}.
\begin{figure} 
\begin{center}
  \includegraphics[scale=0.8]{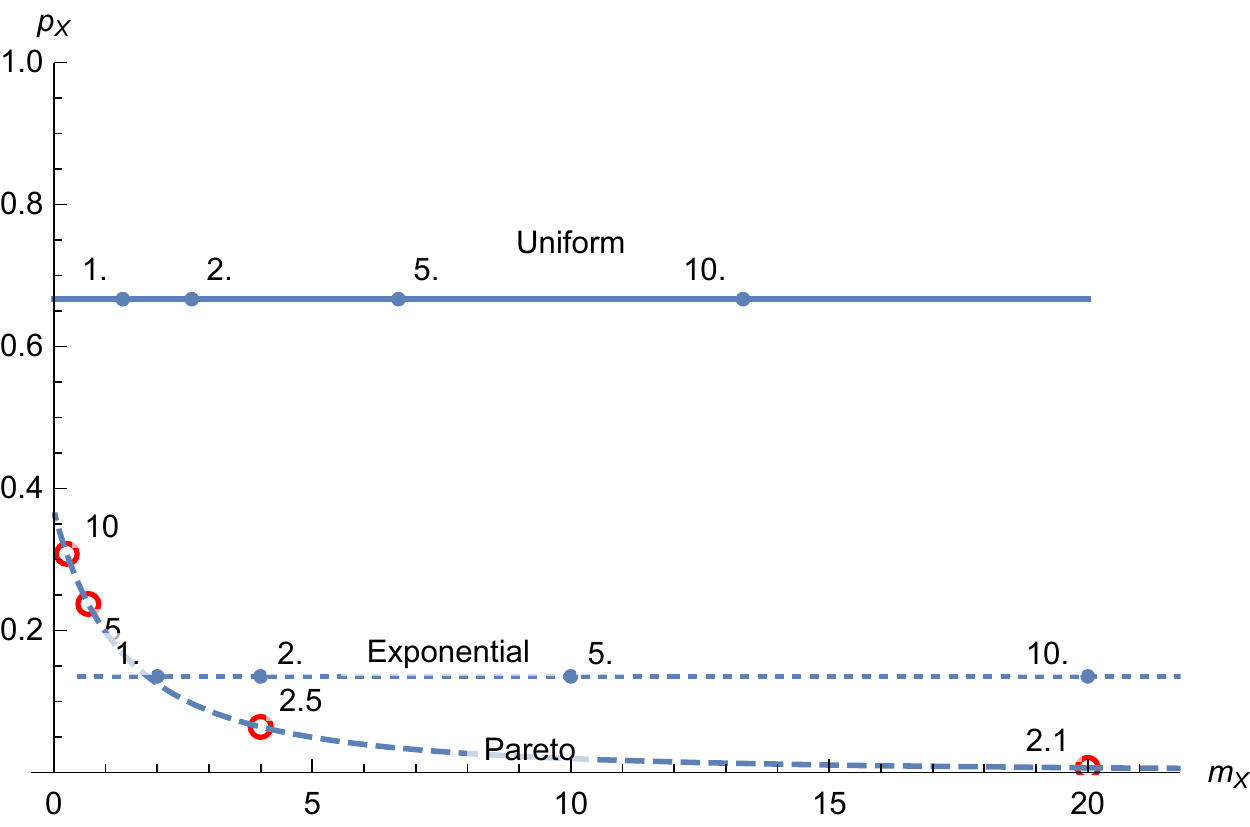}
	\caption{Magnitude-propensity plots for the uniform $U_{[0,a]}$ distribution (solid line), Exponential $Exp(\lambda)$ (dotted) and Pareto $Pa(\theta)$ (dashed), for varying values of the parameter $a, \lambda, \theta$.}
	\label{fig:mp ex1}
	\end{center}
\end{figure}
Both the uniform and exponential distribution have constant propensity along the parameter, and a magnitude linearly increasing with the mean of the distribution. The blue dots show the $(m_X,p_X)$ points for a distribution with mean $\E[X]=1$, $2$, $5$, and $10$: the uniform distribution has a higher propensity, with a lesser magnitude effect.
For the Pareto distribution, the red circles show the $(m_X,p_X)$ points with parameter value $\theta=2.1$, $2.5$, $5$, $10$: the tail effect is reflected in the behavior of the magnitude-propensity pair. For a heavy-tailed distribution ($\theta$ close to $2$), one has a large magnitude with a very small propensity, while for a short-tailed distribution ($\theta$ large), the magnitude effect remain limited but with a larger propensity. This behaviour is in accordance with what could be expected from intuition.

\subsection{Distributions with no closed form expressions for $(m_X,p_X)$}
For some more complicated distributions, one may not have a closed form expression of the $(m_X,p_X)$ as in Examples \ref{ex:uniform}-\ref{ex:pareto}. However, one can use the  computation capabilities of, e.g.,  Mathematica \cite{Mathematica} to obtain a closed-form expression of the $\E [X|X>a]$ and then find numerically the root of (\ref{eq:condition theorem}). We illustrate the procedure for the Gamma and Weibull distributions, two distributions frequently encountered in insurance mathematics.

\begin{example}[Gamma distribution]
\label{ex:gamma}
    For $X\sim \Gamma(\alpha,\beta)$, its density writes $f(x)=x^{\alpha-1}e^{-(x/\beta)}$.  One has 
    $$\E [X|X>a]=\beta\frac{\Gamma(1+\alpha,a/\beta)}{\Gamma(\alpha,a/\beta)}
    $$
    where $\Gamma$ is the incomplete Gamma function. We take a family of $\Gamma(\alpha,2)$ distribution, with shape parameter $\alpha$ varying from $0.1$ to $2.9$ by stepsize of $0.2$. We solve (\ref{eq:condition theorem}) numerically  with initial point $\E [X]$. The resulting magnitude-propensity plot is displayed in Figure \ref{fig:mp gamma}.
    \begin{figure} [H]
    \begin{center}
  \includegraphics[scale=0.8]{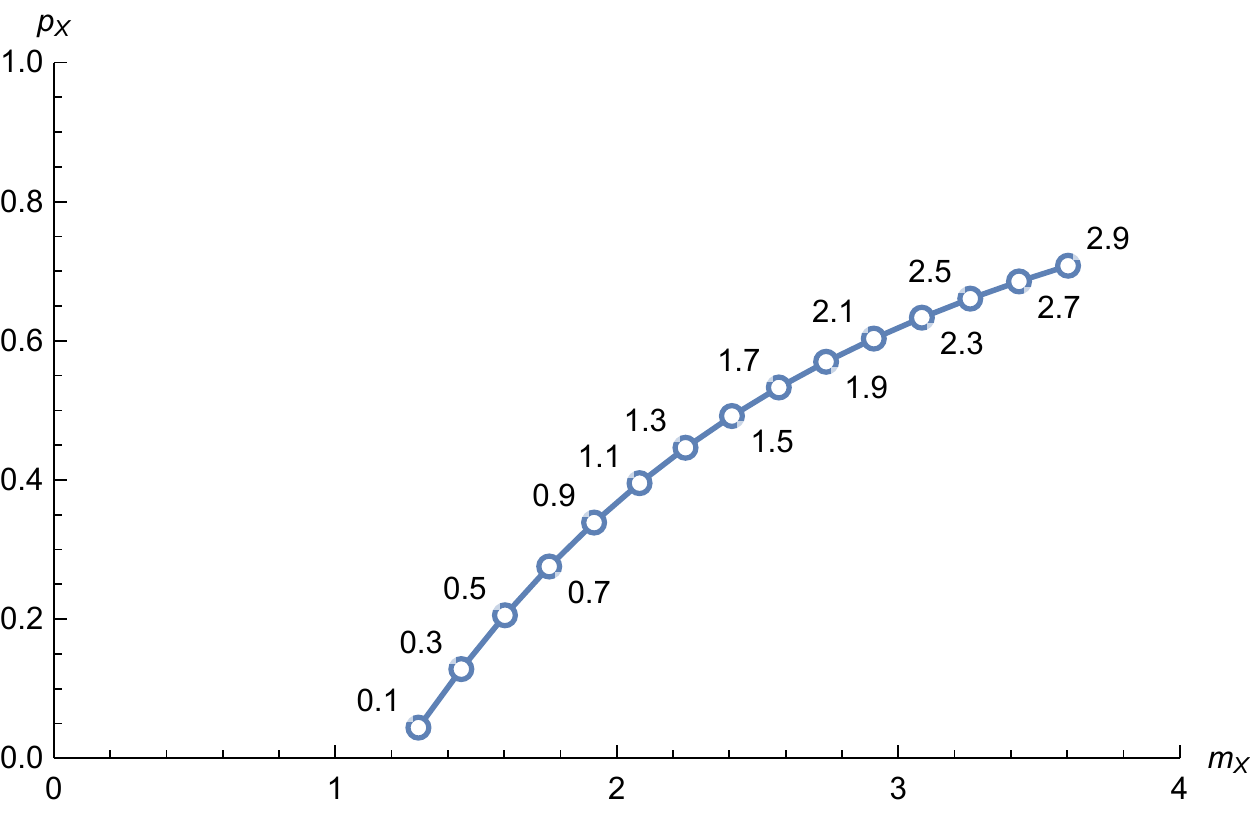}
	\caption{Magnitude-propensity plot  for Gamma distribution $X\sim \Gamma(\alpha,2)$, for $\alpha$ varying from  $0.1$ to $2.9$. The corresponding mean $\E[X]$ is indicated near the circled $(m_X,p_X)$ points.}
	\end{center}
	\label{fig:mp gamma}
\end{figure}
The coordinates of the circled points indicate the magnitude-propensity $(m_X, p_X)$ values and are labeled with the expectation $\E [X]$. Here, it is interesting that both $m_X$ and $p_X$ are increasing with the mean parameter, contrary to the Pareto case of Figure \ref{fig:mp ex1}. This corresponds to intuition: let us recall that as the shape parameter $\alpha$ increase, the density $f(x)$ is shifted to the right and  shrinks toward zero for small $x$ values, see Figure \ref{fig:gammapdf}. Hence, both the magnitude and propensity coding the risk borne by $X$ will increase.
   \begin{figure}[H]
   \begin{center}
  \includegraphics[scale=0.6]{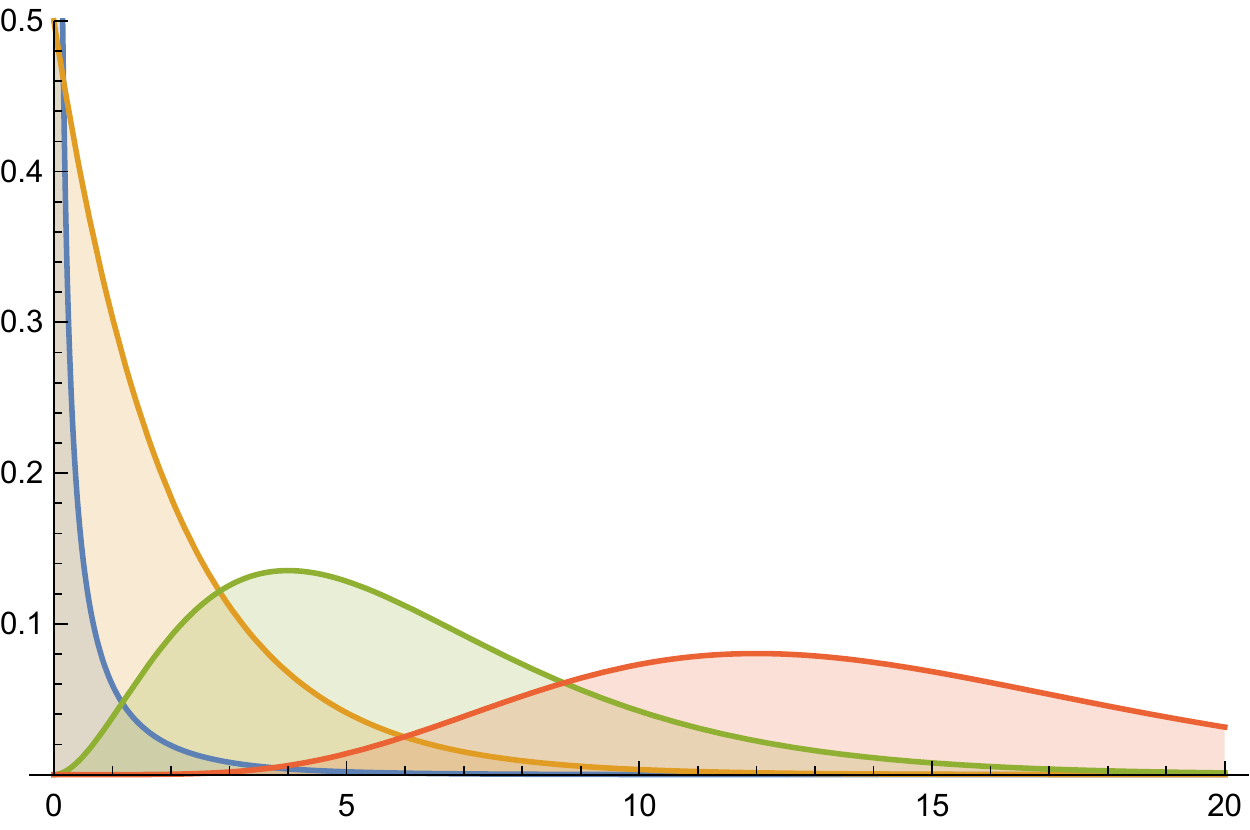}
	\caption{Probability density functions $f$  for Gamma distributions $\Gamma(\alpha,2)$, with $\alpha=0.1$ (blue), $1$ (orange), $3$ (green), $7$ (red). }
	\label{fig:gammapdf}
	\end{center}
\end{figure}
\end{example}

\begin{example}[Weibull distribution]
\label{ex:weibull}
    For $X\sim W(\alpha,\beta)$, its density writes $f(x)=x^{\alpha-1}e^{-(x/\beta)^\alpha}$.  One has 
    $$\E [X|X>a]=\beta e^{a^\alpha \beta^{-\alpha}}\Gamma(1+1/\alpha,a^\alpha\beta^{-\alpha})
    $$
    where $\Gamma$ is the incomplete Gamma function. The effect of shifting the parameters is illustrated in Figure \ref{fig:gammapdf2} below.
    
     \begin{figure}[H]
     \begin{center}
           \includegraphics[scale=0.6]{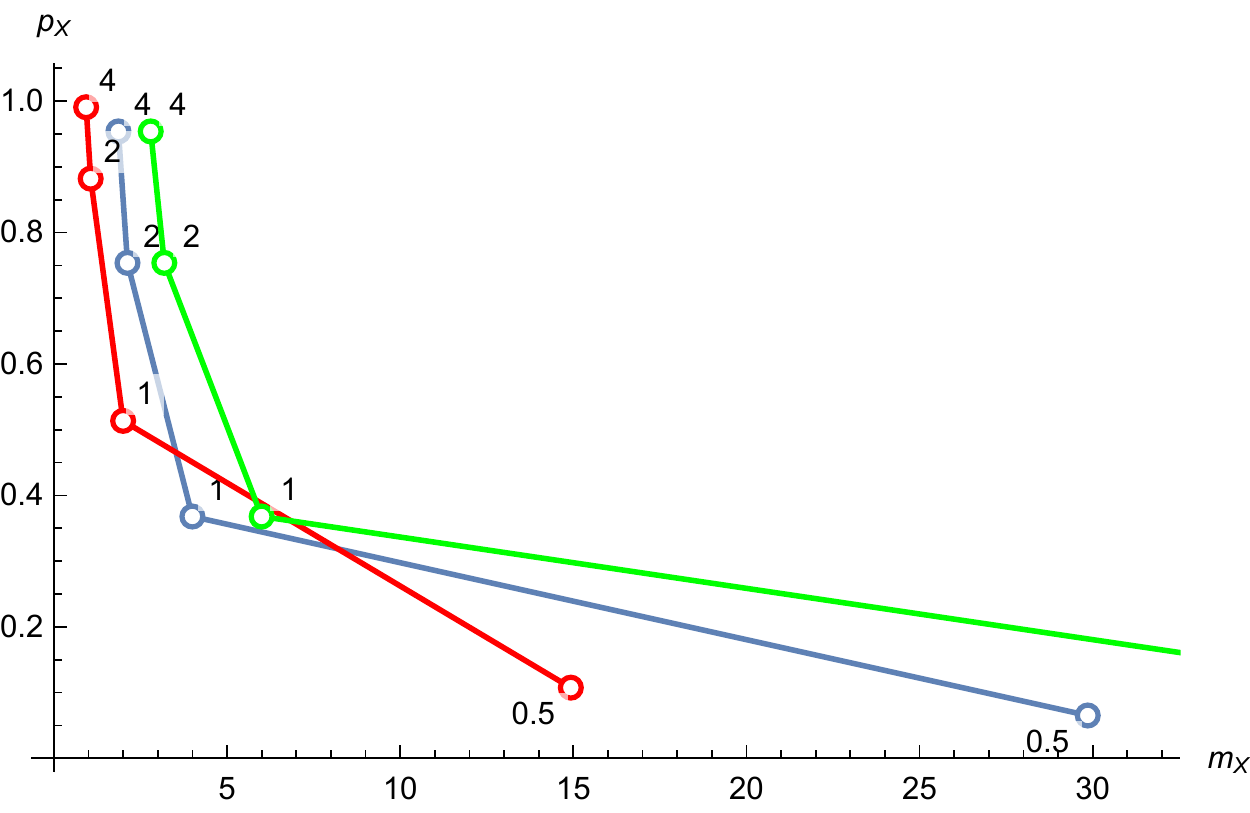}
	\caption{Magnitude-propensity plot  for a Weibull $W(\alpha,\beta)$ distribution for $\beta=2$ (blue), $\beta=1.5$ (red), $\beta=3$ (green). The magnitude-propensity $(m_X, p_X)$ points  are labeled by the values of $\alpha=0.5$, $1$, $2$, $4$.}
	\label{fig:gammapdf2}
     \end{center}
\end{figure}
\end{example}


\subsection{Empirical computations and illustrations}
The characterization (\ref{eq:condition theorem}) shows that the optimal threshold is a fixed point of the function
\begin{equation}
 a\mapsto \frac{\E[X|X\ge a]}{2}.
\label{eq:fixed}   
\end{equation}
 For empirical data, i.e. when the distribution of $X$ is unknown but one has a sample of realizations of $X$, one can replace the unknown expectation in (\ref{eq:fixed}) by a sample estimator. For the empirical measure, one obtains an 
estimate of $(m_X,p_X)$ as a fixed point of $$
a\mapsto \frac{\sum_{i=1}^n X_i\mathds 1_{X_i>a}}{2\sum_{i=1}^n \mathds 1_{X_i>a}}. 
$$
A Mathematica code is given in Appendix \ref{sec:code}.
The latter would correspond to   Lloyd's algorithm in the unconstrained quantization problem. See \cite{Pages2018} for discussion of numerical optimal quantization. 
Alternatively, one could seek directly for a (global) minimizer of the objective function $L$ of (\ref{eq:objective}), using a generic minimizer program (e.g. the command ``NMinimize'' in \cite{Mathematica}).

These approaches are illustrated on a synthetic and   a real dataset. For the synthetic data set, we simulated a sample of $100$ i.i.d. Uniform on $[0,1]$ random variables. From Example \ref{ex:uniform}, we known that we should obtain $(m_X,p_X)=(2/3,2/3)$. LLoyd's algorithm on the sample gives $$(m_X,p_X)\approx (0.67,0.65).$$ NMinimize, which finds for a global minimum, also finds the same $m_X\approx 0.67$.

For the real data set, we take the US Hurricane Losses data, which  is an example data set in \cite{Mathematica}. It reports the thirty most destructive hurricanes in the U.S., from $1949$ to $1999$, see Figure \ref{fig:hurricane}.
\begin{figure} [H]
\begin{center}
      \includegraphics[scale=0.6]{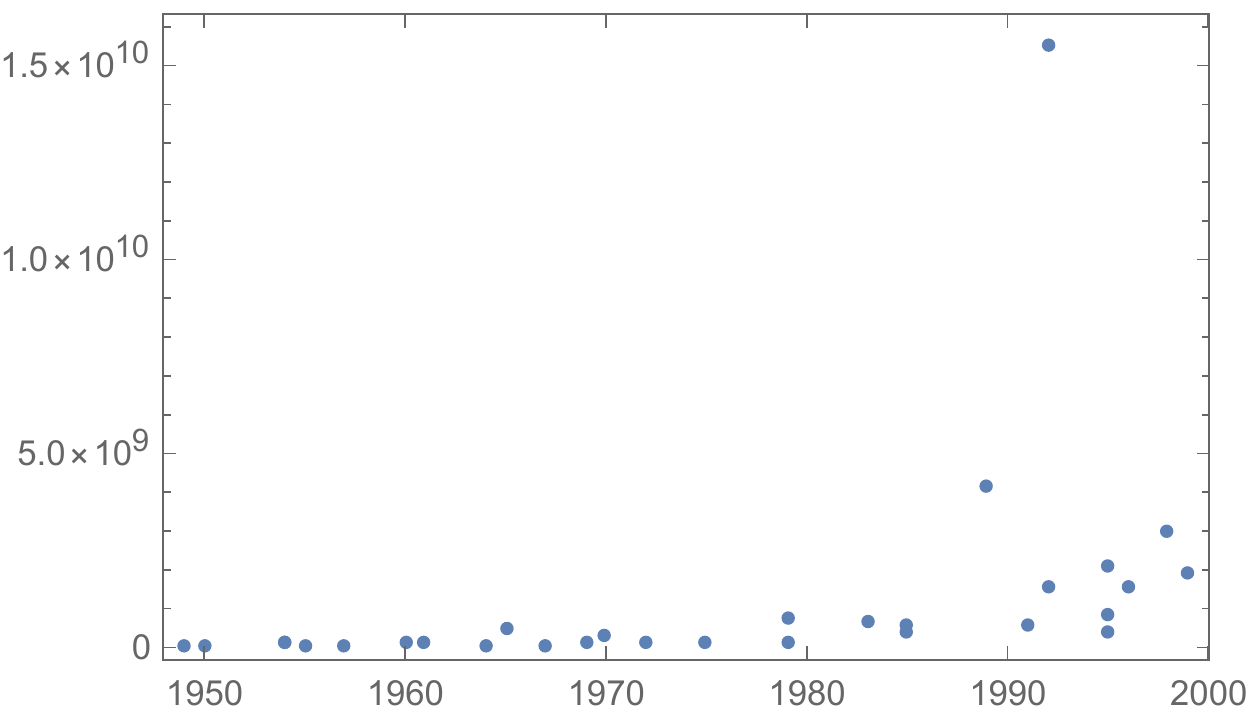}
	\caption{US Hurricane Losses data (Reported losses) in the U.S. $1949--1999$. Available at 
    https://datarepository.wolframcloud.com/resources/US-Hurricane-Loss.
}
\end{center}
	\label{fig:hurricane}
\end{figure}
We standardized the original data by dividing each entry by $10^6$. There is one clear outlier at $15500$ (the famed ``Andrew''  hurricane). It is noteworthy that Lloyd's algorithm on the full dataset give
$$(m_X,p_X)\approx (15500,0.0333),$$
i.e. the outlier value as magnitude, with propensity one over the number of sample points ($1/30$): the outlier has a dwarfing effect on the whole data set and this is reflected on the magnitude-propensity measure.
If the outlier value is removed, one gets
$$(m_X,p_X)\approx (2401.67, 0.206897),$$
which better reflects the magnitude and propensity effects of the data set.

\section{Conclusion}
\label{sec:conclusion}
\subsection{Summary}
In this paper, we introduce the magnitude-propensity risk measures $(m_X, p_X)$, as a way to simultaneously   quantify  both  the  severity and  the  propensity of a risk $X$. Introduced as a particular  mass transportation problem in the Wasserstein metric $W_2$ of the law  of  $X$ to a two-points $\{0,m_X\}$ 
discrete distribution with mass $p_X$ at $m_X$, it generalizes $M-$functionals, in particular the traditional risk measures like Var and Expected Shortfall. The key observation is to view the proposed approach as a constrained optimal quantization problem, with a fixed null center. This allows to parametrize the problem with a single parameter, the optimal threshold determining the Voronoi regions, and to derive it as a solution of a fixed point equation. General existence and characterization results are obtained using B-derivatives, and sufficiency conditions for uniqueness are given.
The obtained magnitude $m_X$ has interesting interpretations in terms of classical risk measures, as  (twice) a Var or as an Expected Shortfall. In addition, it has noteworthy properties, like  positive homogeneity,  monotonicity w.r.t. convex order and being larger than the mean. 

Visualization and comparison of risks can be done on magnitude-propensity plots, which allow for an informative comparison of risks. The effect of tails, shift in the density  and outliers is reflected in the $(m_X,p_X)$ pair. Empirically, the pair can be estimated e.g. by using a variant of Lloyd's algorithm for optimal quantization, or by direct global minimisation.
This novel paradigm of visualizing and comparing risk on the bivariate magnitude-propensity scale offers a broader perspective on risk assessment and evaluation.

\subsection{Perspective: towards general risk quantization}
\label{sec:risk quantization}
To make the presentation   clear,  we focus  in this paper on the simplest way to quantify risk on the magnitude and propensity scale, i.e. on the basic idea (\ref{def})  with a constrained two-points quantizer $\{0,m_X\}$.
This view of measuring risk via a constrained  optimal quantization problem naturally suggests to consider several variants and extensions.  To stimulate further research on the topic, we conclude the paper with a brief sketch below of these possible variants and extensions. 

In all these variants, the rationale is to have a discrete proxy which summarizes the distribution of $X$ and has an intuitive interpretation.
The question of which variant is more sensible from a risk perspective is somehow partially a subjective issue, and is thus left to the appreciation of the reader for the application at hand. What matters is that risk comparisons between several risks be performed within the same framework.

\begin{enumerate}
    \item [(a)] Variant: Mean standardization / Moderate-Large risks.
    
In view of the fact that
$$
W_2^2(P^X,P^Y)=W_2^2(P^{X-\E[X]},P^{Y-\E[Y]})+(\E[X]-\E[Y])^2,
$$
it would make sense to have a discrete proxy $Y$ which has the same mean as $X$.

 One could also standardize differently by the mean, by simply subtracting it. More precisely, one could quantize $X-\E[X]$ to a two-point distributions, $P^Y(.)=p\delta_{m_1}(.)+(1-p)\delta_{m_2}(.)$, with $m_1<m_2$. This allows to summarize the distribution of $X$ into a mean effect $\E[X]$, which itself decomposes into a ``moderate risk'' of magnitude $m_1+\E[X]$ and propensity $p_1$, and a ``large'' risk of magnitude $m_2+\E[X]$ and propensity $p_2$. This gives a quintuplet of descriptive statistics, which summarizes the characteristics of the distribution and can be thought as an alternative to Tukey's box plot.
 
 \item [(b)] Variant: three-points quantification.
 
 Also, one can refine our proxy by looking for a quantization to a discrete distribution with more than two points. For example, another way to  refine one's measure of risk   into a ``moderate'' risk and a ``large'' risk (or, say, ``Low'' risk and ``Tail'' risk), is to use directly a three points discrete measure, $P^Y=(1-p_1-p_2)\delta_0+p_1\delta_{m_1}+p_2\delta_{m_2}$, with $m_1<m_2$.
 With a such three points discrete measure,  one can  encode and quantify  both moderate, resp.  large risk,   in   the magnitude and propensity scale with  $(m_1,p_1)$, resp. $(m_2,p_2)$.

 Such a simultaneous quantization  of the ``mild'' and ``extreme'' risk parts of an insurance risk $X$  could reveal to be very  valuable in the field of Reinsurance, see e.g. \cite{albrecher2017reinsurance}.    The insurer and reinsurer have to agree on an optimal risk sharing policy, with the insurer usually ceding to the reinsurer the ``extreme'' part of the risk (the one with high magnitude and low propensity), while keeping and managing the ``mild'' part. The proposed risk quantization approach by magnitude and propensity seems to be particularly well-suited for this task.
 
 \item [(c)] Extension to financial risk.
 
 For $X$ real-valued (i.e. we take the financial mathematics convention, with $X\ge 0$ standing for a gain and $X<0$ for a loss), one can consider for example  a three-points risk quantization with the following distribution,
$$P^Y(.)=p_-\delta_{m_-}(.) +(1-p_+-p_-)\delta_{\E[X]}(.)+p_+\delta_{m_+}(.),$$
with $m_-\le \E[X]\le m_+$
or with
$$P^Y(.)=p_-\delta_{m_-}(.) +(1-p_+-p_-)\delta_{0}(.)+p_+\delta_{m_+}(.),$$
with the constraint $\E [X]=\E[Y]$, $m_-\le 0\le m_+$.
This allows to summarize the gain/loss of the distribution on both the magnitude and propensity scales.

\item [(d)] Multivariate risk.

Multivariate generalisations to  risk vectors $\mathbf X\in\mathbb R^d$ are similar in spirit, although the computations are less explicit. One can somehow reduce to the scalar case by considering the risk quantization of a portfolio vector $\boldsymbol{\beta}^T \mathbf X$, for $\boldsymbol{\beta}$ in unit simplex,  see e.g.  \cite{rockafellar2002conditional}, \cite{UR2000}.

\item [(e)] Covariates.

Another extension is to take into account the effect of covariates $\mathbf X$ on a loss variable $Y$,  by quantizing the risk of the conditional distribution $Y|\mathbf X$, (or a linear approximation thereof, as in quantile regression). Note that a related, but distinct, approach occurs in Frequency-Severity models (see e.g. \cite{Frees2010} Chapter 16): there, the loss distribution has a large proportion of zeros, and the modeling is done in two parts, one for the frequency of zeros, and the other for the severity. Among other, models include the Tobit model \cite{tobin1958}, with a censored latent variable and the individual risk model, see \cite{Frees2010}. 
\end{enumerate}

\section*{Acknowledgments}
Olivier P. Faugeras acknowledges funding from ANR under grant ANR-17-EURE-0010 (Investissements d’Avenir program).

\noindent{{\it Declaration of interest:}  none.
}

\section*{Appendix: Mathematica code}
\label{sec:code}
Empirical computation of $(m_X,p_X)$   by Lloyd's Method:
\begin{verbatim}
  (*Input: data= dataset; 
          x0=starting value of the threshold search (take the mean
          of the data set for  example)
    Ouput: {m_X,p_X} 
  *)    
    mplloyd[data_, x0_] := 
    Module[{a}, 
        iteratefunction[a_] := Mean[Select[data, # > a &]]/2; 
        a = FixedPoint[iteratefunction[#] &, x0];
        {2 a, NProbability[x > a, x \[Distributed]
        EmpiricalDistribution[data]]}
          ]
    (* one could also have computed p_x by
    1-CDF[EmpiricalDistribution[data],a] *)
\end{verbatim}


\printbibliography
\end{document}